\documentclass[11pt]{amsart}

\usepackage[dvipsnames]{xcolor}
\usepackage{todonotes}
\usepackage[utf8]{inputenc}
\usepackage{comment,enumerate,graphicx}
\usepackage{enumitem,amsmath,amssymb}
\usepackage{amsthm}
\usepackage{thmtools}
\usepackage{url}
\usepackage{caption}
\usepackage{subcaption}
\usepackage[top=23 mm, bottom=23 mm, left=22 mm, right= 22 mm]{geometry}
\usepackage{relsize}
\usepackage{microtype}

\usepackage[hidelinks]{hyperref}
\makeatletter
\newcommand{\labelinthm}[1]{%
   \label{temp#1}
   \protected@write \@auxout {}{\string \newlabel{#1}{{\emph{\ref{temp#1}}}{\thepage}{\emph{\ref{temp#1}}}{temp#1}{}} }%
}
\makeatother
\usepackage{adjustbox}
\usepackage{tikz}
\usetikzlibrary{math,calc,intersections, scopes}
\usetikzlibrary{positioning,arrows,shapes,decorations.markings,decorations.pathreplacing, decorations.pathmorphing,matrix,patterns}
\tikzstyle{vertex}=[circle,draw=black,fill=black,inner sep=0,minimum size=5pt,text=white,font=\footnotesize]
\tikzstyle{redvertex}=[circle,draw=red,fill=red,inner sep=0,minimum size=5pt,text=white,font=\footnotesize]
\definecolor{amber}{rgb}{1.0, 0.75, 0.0}
\definecolor{darkgreen}{rgb}{0.18, 0.7, 0.46}

\declaretheorem[name=Theorem]{theorem}

\newtheorem{lemma}[theorem]{\bf Lemma}

\newtheorem*{theorem*}{\bf Theorem}

\theoremstyle{definition}

\newcommand\claimproofend{\renewcommand{\qedsymbol}{$\boxdot$}
\end{proof}
\renewcommand{\qedsymbol}{$\square$}
}
\def\eps{\varepsilon}

\def\cA{\mathcal{A}}
\def\cB{\mathcal{B}}

\def\cF{\mathcal{F}}
\def\cG{\mathcal{G}}

\def\bE{\mathbb{E}}

\def\Pb{\mathbb{P}}

\newcommand{\DISJ}{\mathrm{DISJ}}
\newcommand{\SD}{\mathrm{SD}}

\title{\vspace{-0.9cm} Communication Complexity of Disjointness under Product Distributions}
\date{}

\author{Zach Hunter\textsuperscript{1}}
\author{Aleksa Milojevi\'c\textsuperscript{1}}
\author{Benny Sudakov\textsuperscript{1}}
\thanks{\textsuperscript{1}Department of Mathematics, ETH Z\"urich, Switzerland. Email: {\tt \{zach.hunter, aleksa.milojevic, benjamin.sudakov\}@math.ethz.ch}. Research supported in part by SNSF grant 200021-228014.}
\author{Istvan Tomon\textsuperscript{2}}
\thanks{\textsuperscript{2}Ume\r{a} University, \emph{e-mail}: \texttt{istvantomon@gmail.com}, Research supported in part by the Swedish Research Council grant VR 2023-03375.}

\begin{document}

\begin{abstract}
Determining the randomized (or distributional) communication complexity of disjointness is a central problem in communication complexity, having roots in the foundational work of Babai, Frankl, and Simon in the 1980s and culminating in the famous works of Kalyanasundaram-Schnitger and Razborov in 1992. However, the question of obtaining tight bounds for product distributions persisted until the more recent work of Bottesch, Gavinsky, and Klauck resolved it. In this note we revisit this classical problem and give a short, streamlined proof of the best bounds, with improved quantitative dependence on the error parameter. 

Our approach is based on a simple combinatorial lemma that may be of independent interest: if two sets drawn independently from two distributions are disjoint with non-negligible probability, then one can extract two subfamilies of reasonably large measure that are fully cross-disjoint (equivalently, a large monochromatic rectangle for disjointness).  
\end{abstract}

\maketitle

\section{Introduction}

Communication complexity studies how much information two separated parties must exchange in order to compute a function of their joint input.
Formally, fix finite sets $\cA,\cB$ and a Boolean function $f:\cA\times\cB\to\{0,1\}$.
Alice receives $A\in\cA$ and Bob receives $B\in\cB$, and neither party sees the other's input.
They communicate by exchanging bits across a two-way channel according to a (deterministic) protocol, at the end of which one of the players outputs a value intended to equal $f(A,B)$.
The \emph{cost} of a protocol is the maximum, over all inputs $(A,B)$, of the total number of bits exchanged on input $(A,B)$, and the \emph{deterministic communication complexity} of $f$, denoted by ${\rm CC}(f)$, is the minimum cost of a protocol that is correct on every input. We refer to \cite{KN, RY} for background, motivation, and applications.
	
One of the most studied functions is \emph{set disjointness}. Fix a ground set $[n]$ and let $\cA=\cB=2^{[n]}$ and
\[ \DISJ_n(A,B)=1 \quad\Longleftrightarrow\quad A\cap B=\varnothing.\]

Disjointness is the canonical example of a hard function to compute in this model and it is a starting point for many reductions. Its deterministic communication complexity is $\Theta(n)$, and the lower bound follows from the rank method by noting that ${\rm CC}(\DISJ_n)\geq \log_2 {\rm rank} M$, where $M_{A, B}=\DISJ_{n}(A, B)$ is a $2^n\times 2^n$ matrix of full rank which encodes the function $\DISJ_n$ (see \cite{S} for details).

In this paper we study disjointness in the \emph{distributional} model.
Fix a distribution $\mu$ on $\cA\times\cB$ and an error parameter $\eps\in(0,1/2)$. In this model, the input $(A, B)$ is sampled according to $\mu$, and $D^\mu_\eps(f)$ denotes the minimum communication cost of a deterministic protocol that computes $f$ correctly on a $(1-\eps)$-fraction of pairs of inputs (with respect to $\mu$). Note that the protocol may depend on $\mu$. 

By Yao's minimax principle, $\max_{\mu} D^\mu_\eps(f)$ equals the public-coin randomized communication complexity of $f$ with error~$\eps$. Therefore, determining the distributional complexity of disjointness was a problem of great interest, until it was famously resolved by Kalyanasundaram and Schnitger \cite{KS92} and by Razborov \cite{R92} in the 1990s. Razborov proved that disjointness has randomized communication complexity $\Omega(n)$ by constructing a hard distribution $\mu$ for which $D^\mu_{1/3}(\DISJ_n)=\Theta(n)$. Thus, in the fully general distributional/randomized setting disjointness remains linearly hard.

It is natural to consider the restricted setting when $\mu=\mu_A\times \mu_B$ is a \textit{product distribution}. This corresponds to sampling the sets $A\sim \mu_A, B\sim \mu_B$ independently. Following the classical work of Babai, Frankl, and Simon \cite{BFS86}, we define the \textit{strong distributional complexity} of a function $f$ as
\[ \mathrm{SD}_\eps(f):=\sup_{\mu=\mu_A\times \mu_B} D^\mu_\eps(f). \]

Remarkably, restricting to independent inputs decreases the worst-case complexity of disjointness from linear to about $\sqrt n$. More precisely, Babai--Frankl--Simon \cite{BFS86} proved in 1986 that for every $\eps<1/2$,
\[ \Omega(\sqrt n)\leq\SD_\eps(\DISJ_n)\leq O_\eps(\sqrt n\log_2 n). \]

The logarithmic gap between the upper and lower bound persisted for several decades, and was closed only in 2015 in the work of Bottesch, Gavinsky, and Klauck \cite{BGK15}, who proved the tight bound $\SD_\eps(\DISJ_n)=\Theta_\eps(\sqrt n)$.

They also considered this problem in a more general context, which interpolates between the product-distribution regime and the fully general regime by parameterizing distributions according to their mutual information $I(A:B)$. Recall, the \textit{mutual information} between random inputs $A$ and $B$ sampled from the distribution $\mu$ with marginals $\mu_A, \mu_B$ is 
\[I(A:B)=\sum_{a\in \cA}\sum_{b\in \cB}\mu(a, b)\log_2\frac{\mu(a, b)}{\mu_A(a)\mu_B(b)}.\]

More generally, they showed that if $\mu$ is a distribution with mutual information at most $k$, then $D_{\eps}^\mu(\DISJ_n)\leq O(\sqrt{n(k+1)}/\eps^2)$. The proofs of these results rely crucially on the H\aa stad--Wigderson protocol \cite{HW} which tests for disjointness of $\ell$-element sets inside an arbitrary ground set using only $O_\eps(\ell)$ bits.

In this paper we revisit this classical problem and give a simple streamlined proof of the result of Bottesch, Gavinsky, and Klauck~\cite{BGK15}. Additionally, this proof  improves the dependence on the error parameter $\eps$.

\begin{theorem}\label{thm:main_MI}
Let $\mu$ be a distribution on $2^{[n]}\times 2^{[n]}$ and let $(A,B)\sim\mu$ satisfy $I(A:B)\leq k$.
Then \[ D^\mu_\eps(\DISJ_n)\leq O\left(\sqrt{\frac{n(k+1)}{\eps}}\right). \]
\end{theorem}

Interestingly, Bottesch, Gavinsky and Klauck provide a lower bound which shows that Theorem~\ref{thm:main_MI} is tight, at least when the mutual information is not very large. Namely, for every $1\leq k\leq \eps n$ they construct a measure $\mu$ with mutual information at most $k$, for which $D^\mu_\eps(\DISJ_n)\geq \Omega(\sqrt{n(k+1)/\eps})$, showing that the above bound is tight in this regime.

The special case when $k=0$ corresponds to the case of product distributions, and there we obtain a very short and completely self-contained proof of the bounds for the strong distributional complexity of disjointness.

\begin{theorem}\label{thm:main_product}
For any $\eps<1/2$,
\[\SD_\eps(\DISJ_n)\leq O\left(\sqrt{n\log_2(1/\eps)}\right).\]
\end{theorem}

The protocol we construct is deterministic, and it is based on a combinatorial lemma, which we believe may be of independent interest. Roughly speaking, the lemma says if two independent sets $A\sim\mu_A$ and $B\sim\mu_B$ are disjoint with non-negligible probability, then one can extract two subfamilies $\mathcal{A},\mathcal{B}\subseteq 2^{[n]}$ of reasonably large measure such that $A\cap B=\varnothing$ for every pair $(A,B)\in\mathcal{A}\times\mathcal{B}$ (equivalently, $\DISJ_n$ is identically $1$ on a large combinatorial rectangle). Using this lemma and a variant of the Nisan-Wigderson reduction, we can use large monochromatic combinatorial rectangles to obtain an efficient protocol for disjointness (see Section~\ref{sec:reduction} for details).

The proof of this lemma is self-contained and is inspired by ideas developed in our recent work on the Singer--Sudan conjecture concerning families with many disjoint pairs~\cite{HMST}. Namely, inspired by finding monochromatic rectangles and the log-rank conjecture, Singer and Sudan \cite{SS22} asked the following question: given families $\cA, \cB \subseteq 2^{[n]}$ with the property that at least $\eps |\cA||\cB|$ pairs $(A, B)\in \cA\times \cB$ are disjoint, how large families $\cF\subseteq \cA, \cG\subseteq \cB$ can we find with the property that every pair of sets $(F, G) \in \cF \times \cG$ is disjoint? In \cite{HMST}, we showed one can take $|\cF||\cG|\geq 2^{-O(\sqrt{n \log_2 1/\eps})}|\cA||\cB|$. Our main lemma is a generalization of this statement, where set families are replaced by distributions over sets.

\medskip\noindent
\textbf{Paper organization.} In Section~\ref{sec:disjoint}, we state and prove our main lemma. Then, in Section~\ref{sec:reduction} we use it to prove Theorem~\ref{thm:main_product}. Finally, we use Theorem~\ref{thm:main_product} in Section~\ref{sec:mutual} to derive Theorem~\ref{thm:main_MI}.

\section{Distributions with many disjoint sets}\label{sec:disjoint}

\begin{lemma}\label{lemma:main}
Let $\mu_A, \mu_B$ be probability distributions on $2^{[n]}$, and let $\eps\in (2^{-n}, 1/2)$ be a parameter such that  \[\Pb_{A\sim \mu_A, B\sim \mu_B}[A\cap B=\varnothing]\geq \eps.\] Then, there exist collections $\cA, \cB\subseteq 2^{[n]}$ such that $\mu_{A}(\cA), \mu_B(\cB)\geq 2^{-O(\sqrt{n\log_2 1/\eps})}$, and such that all pairs $(A, B)\in \cA\times \cB$ are disjoint.
\end{lemma}
\begin{proof}
The main idea of the proof is to pick $\ell=\lfloor \sqrt{n/\log_2\eps^{-1}}\rfloor$ sets $A_1, \dots, A_\ell$ randomly according to $\mu_A$ and consider their union $U=\bigcup_{i=1}^\ell A_i$. Then set $\cA$ to be the family of sets in $2^{[n]}$ contained in $U$ and set $\cB$ to be the family of sets in $2^{[n]}$ contained in $[n]\backslash U$. This construction ensures that for each pair of sets $(A,B)\in \cA\times \cB$, $A$ and $B$ are disjoint, and therefore our main goal is to show that with positive probability, both $\mu_A(\cA)$ and $\mu_B(\cB)$ are sufficiently large.

First, we show that \[\Pb[\mu_A(\cA)\geq 2^{-3n/\ell}]\geq 1-2^{-2n}.\] Say that a set $U$ is \textit{bad} if $\Pb_{A\sim \mu_A}[A\subseteq U]=\mu_A(2^U)\leq 2^{-3n/\ell}$. Then the probability that $\ell$ randomly chosen sets from $\mu_A$ land inside a bad set $U$ is at most $(2^{-3n/\ell})^\ell=2^{-3n}$. Hence, taking the union bound over at most $2^n$ bad sets, the probability that $U=\bigcup_{i=1}^\ell A_i$ is bad can be bounded by $2^{-2n}$. Therefore, with probability at least $1-2^{-2n}$, we have $\mu_A(\cA)\geq 2^{-3n/\ell}$.

Next, we show that \[\Pb[\mu_B(\cB)\geq \eps^\ell/2]> 2^{-n}.\]
To do so, we shall compute the expected measure of $\cB$ by summing over all sets $B\subseteq [n]$ the probability that $B$ lands in $\cB$ times the weight of $B$ under $\mu_B$. By observing that \[\Pb\big[B\in \cB\big]=\Pb\big[B\cap \bigcup_{i=1}^\ell A_i=\varnothing\big]=\Pb\big[A_i\subseteq [n]\backslash B\text{ for all }i\in \{1, \dots, \ell\}\big]=\mu_A(2^{[n]\backslash B})^\ell,\] we obtain
\[\bE\big[\mu_B(\cB)\big]=\sum_{B\subseteq [n]} \mu_B(B)\cdot \Pb\big[B\in \cB\big]=\sum_{B\subseteq [n]} \mu_B(B)\cdot \mu_A(2^{[n]\backslash B})^\ell.\]
By applying Jensen's inequality to the convex function $x\mapsto x^\ell$ and the measure $\mu_B$ we obtain
\begin{align*}
\bE\big[\mu_B(\cB)\big]&=\sum_{B\subseteq [n]} \mu_B(B)\cdot \mu_A(2^{[n]\backslash B})^\ell\geq \Big(\sum_{B\subseteq [n]} \mu_B(B)\cdot \mu_A(2^{[n]\backslash B})\Big)^\ell\geq \eps^\ell.
\end{align*}
The last inequality follows from the fact $\sum_{B\subseteq [n]} \mu_B(B)\cdot \mu_A(2^{[n]\backslash B})=\Pb_{A\sim \mu_A, B\sim \mu_B}[A\cap B=\varnothing]\geq \eps$. Using that $0\leq \mu_B(\cB)\leq 1$, we can write \begin{equation}\label{eq: common neighborhood}\eps^\ell \leq \bE\big[\mu_B(\cB)\big]\leq  1\cdot \Pb\Big[\mu_B(\cB)\geq \eps^\ell /2\Big] + \frac{\eps^\ell }{2}\cdot \Pb\Big[\mu_B(\cB)< \eps^\ell/2\Big] <  \Pb\Big[\mu_B(\cB)\geq \eps^\ell /2\Big]+\frac{\eps^\ell }{2}.\end{equation}
Note that $\eps^\ell/2\geq \eps^{\sqrt{n/\log_2\eps^{-1}}}/2 = 2^{-\sqrt{n\log_2\eps^{-1}}-1}> 2^{-n-1}$, where the last inequality holds by our assumption $\eps> 2^{-n}$. Therefore, comparing the two sides of Eq.~\ref{eq: common neighborhood} gives $\Pb[\mu_B(\cB)\geq \eps^\ell/2]>\eps^{\ell}/2\geq 2^{-n-1}$. 

Since $\Pb[\mu_A(\cA)\geq 2^{-3n/\ell}]+\Pb[\mu_B(\cB)\geq \eps^\ell/2]>1$, there exists a choice of $A_1,\dots,A_\ell$ for which 
\begin{align*}
\mu_A(\cA)&\geq 2^{-3n/\ell}\geq 2^{-O(\sqrt{n\log_2\eps^{-1}})},\\
\mu_B(\cB)&\geq \eps^\ell /2=2^{-\ell\log_2\eps^{-1}-1}\geq 2^{-O(\sqrt{n\log_2\eps^{-1}})},
\end{align*}
where we have used that $\ell=\lfloor \sqrt{n/\log_2\eps^{-1}}\rfloor$. This completes the proof.
\end{proof}

\section{Distributional Complexity for Product Distributions}\label{sec:reduction}

In this section, we use a variant of the Nisan-Wigderson reduction \cite{NW} to prove Theorem~\ref{thm:main_product}, i.e. to prove that $D^{\mu_A\times \mu_B}_\eps(\DISJ_n)\leq O(\sqrt{n\log_2 1/\eps})$. We say that a measure $\nu$ is \textit{almost-empty} if $\Pb_{(A, B)\sim \nu}[A\cap B=\varnothing]\leq \eps/2$. Also, we say that a combinatorial rectangle $R=\cF\times \cG\subseteq 2^{[n]}\times 2^{[n]}$ is \textit{full} if all pairs $(F, G)\in \cF\times \cG$ are disjoint.

\begin{proof}[Proof of Theorem~\ref{thm:main_product}.]
We present an efficient protocol for computing $\DISJ_n$ correctly on a $(1-\eps)$-fraction of the inputs under the given measure $\mu$. Let $A, B\in 2^{[n]}$ be the inputs known to Alice and Bob. 

We first define a subprotocol, with the following goal. Assume that there is a set $X\subset [n]$, known to both Alice and Bob, such that $A\setminus X$ and $B\setminus X$ are disjoint, or equivalently, $A\cap B=\varnothing \Leftrightarrow (A\cap X)\cap (B\cap X)=\varnothing$. At the end of the subprotocol, either Alice and Bob declare whether $A$ and $B$ are disjoint, or they will agree on a set $Y$ of size $|Y|\leq |X|/2$ such that $A\setminus Y$ and $B\setminus Y$ are disjoint, effectively halving the size of the ground set. Additionally, if $T$ is the set of pairs of inputs where the subprotocol outputs an answer whether $A\cap B=\varnothing$ and $S$ is the set of pairs of inputs where the answer was incorrect, we will guarantee $\mu(S)\leq \frac{\eps}{2}\big( \mu(T)+\frac{1}{|X|}\big)$.

We claim that such a subprotocol can be executed with $O(\sqrt{|X|\log_2 1/\eps})$ bits of communication. This claim implies our theorem, since we can start with $X=[n]$, iterate the subprotocol until the size of the ground set reduces to $|X|\leq 4\log_2 1/\eps$, halving the size of the ground set in every step. If up to this point Alice and Bob did not decide whether $A\cap B=\varnothing$, Alice can communicate all elements of $A\cap X$ using $|X|=O(\sqrt{|X|\log_2 1/\eps})$ bits, thus solving the problem. The total number of bits communicated in this protocol is at most
$$\sum_{i=0}^{\lfloor \log_2 n\rfloor-1}O\left(\sqrt{\frac{n}{2^i} \log_2 1/\eps}\right)=O(\sqrt{n\log_2 1/\eps}).$$

Moreover, the measure of the set $S$ where the protocol makes an incorrect answer is at most $$\frac{\eps}{2}\left(1+\sum_{i=0}^{\lfloor \log_2 n\rfloor-1}\frac{1}{n/2^i}\right)\leq \frac{\eps}{2}\cdot 2=\eps.$$

\medskip\noindent
\textbf{Subprotocol specification:} Let $M_0$ be the disjointness matrix of the subsets of $X$, that is, for any $A_0, B_0\subseteq X$ we have $(M_0)_{A_0, B_0}=1 \Longleftrightarrow A_0\cap B_0=\varnothing$. Also, let $A'=X\cap A$, $B'=X\cap B$, $\cA_0=\cB_0=2^X$, and let $\nu=\nu_{A}\times \nu_{B}$ be the rescaling of the measure $\mu$ on $2^X\times 2^X$ such that $\nu$ is a probability measure. 

If $\nu_{A}(A')\leq \frac{\eps}{2^{2|X|}}$, then Alice communicates this to Bob, or if $\nu_{B}(B')\leq \frac{\eps}{2^{2|X|}}$, Bob communicates this to Alice using a single bit. If either of these two events happen, Alice and Bob may declare that the sets are, say, not disjoint. The measure of inputs where they make a mistake in this step is at most $\frac{\eps}{2^{|X|-1}}\leq \frac{\eps}{2|X|}$, since the total measure of pairs $(A', B')$ with $\nu_{A}(A')\leq \frac{\eps}{2^{2|X|}}$ or $\nu_{B}(B')\leq \frac{\eps}{2^{2|X|}}$ is at most $2\cdot 2^{|X|}\cdot {\eps}/{2^{2|X|}}=\eps/2^{|X|-1}$.

If $\nu$ is almost-empty, then Alice and Bob declare that their sets are not disjoint and end the protocol. Then their answer is correct for at least $(1-\eps/2)$ measure of the inputs in $\cA_0\times \cB_0$, and thus $\mu(S)\leq \frac{\eps}{2}\mu(T)+\frac{\eps}{2|X|}$ is satisfied.

If $\nu$ is not almost-empty, let $R=\cF\times \cG$ be the full rectangle in $M_0$ of largest measure. Since $\eps/2\geq 2^{-|X|}$ (due to the stopping condition on $X$), by Lemma~\ref{lemma:main}, we have $\nu(R)\geq 2^{-O(\sqrt{|X|\log_2 1/\eps})}$, which also implies $\nu_{A}(\cF),\nu_{B}(\cG)\geq 2^{-O(\sqrt{|X|\log_2 1/\eps})}$. Since all sets $F\in \cF$ are disjoint from all sets $G\in \cG$, we have that their unions $U=\bigcup_{F\in \cF}F$ and $V=\bigcup_{G\in \cG}G$ are also disjoint. Hence, $|U|+|V|\leq |X|$ and so one of $U, V$ has size at most $|X|/2$, say $U$ (otherwise, we may swap the roles of the players). If $A'\in \cF$, then $A'\subseteq U$, so $A'\setminus U$ and $B'\setminus U$ are disjoint. In this case, Alice and Bob set $Y:=U$ and stop the subprotocol. If $A'\notin \cF$, then define $\cA_{1}:=\cA_0\setminus \cF$ and $\cB_{1}:=\cB_0$.

One can iterate the above procedure, thus defining a sequence of families of subsets $\cA_0\supseteq \cA_1\supseteq \dots\supseteq \cA_t$ and $\cB_0\supseteq \cB_1\supseteq \dots\supseteq \cB_t$. Observe that for each $1\leq i\leq t$ we have $$\nu(\cA_{i}\times \cB_{i})\leq  \big(1-2^{-O(\sqrt{|X|\log_2 1/\eps})}\big)\nu(\cA_{i-1}\times \cB_{i-1}).$$
Therefore, if the subprotocol runs for $t$ steps, we have
$$\nu(\cA_t\times \cB_t)\leq \big(1-2^{-O(\sqrt{|X|\log_2 1/\eps})}\big)^t\leq \exp(-t2^{-O(\sqrt{|X|\log_2 1/\eps})}).$$
Setting $t:=2^{C \sqrt{|X|\log_2 1/\eps}}(|X|+\log_2 1/\eps)=2^{O(\sqrt{|X|\log_2 1/\eps})}$ for some large constant $C$, the right-hand-side is at most $2^{-4|X|-2\log_21/\eps}\leq \frac{\eps^2}{2^{4|X|}}$. But this is impossible as then $\nu(A'\times B')\leq \frac{\eps^2}{2^{4|X|}}$, implying that either $\nu_{A}(A')\leq \frac{\eps}{2^{2|X|}}$ or $\nu_{B}(B')\leq  \frac{\eps}{2^{2|X|}}$. Hence, this subprotocol runs for at most $t$ steps, and it defines a protocol tree with $O(t)$ leaves.

It is well-known (see e.g. Chapter 2, Lemma 2.8 in \cite{KN}) that every protocol tree can be rebalanced to be a decision tree with depth logarithmic in the number of leaves. This implies that the subprotocol can be executed with $O(\log_2 t)=O(\sqrt{|X|\log_2 1/\eps})$ bits of communication, verifying our claim.
\end{proof}

\section{Distributions with Mutual Information}\label{sec:mutual}

Next, we derive Theorem~\ref{thm:main_MI} from Theorem~\ref{thm:main_product}. Bottesch, Gavinsky and Klauck observed that the (classical) Substate theorem of Jain, Radhakrishnan and Sen \cite{JRS02} can be used to translate bounds from the independent setting to the setting of bounded mutual information. However, the main point here is that the bounds of Theorem~\ref{thm:main_product} have a very good dependence on $\eps$, which allows us to apply this reduction without big losses.

To state the Substate theorem, we recall the following two definitions. If $\mu, \nu$ are two probability distributions on $\cA\times \cB$, $\|\mu-\nu\|_{\operatorname{TV}}$ stands for the total variation distance between distributions $\mu$ and $\nu$, which is defined as $\|\mu-\nu\|_{\operatorname{TV}}=\max_{S\subseteq \cA\times \cB}|\mu(S)-\nu(S)|$. Further, if $\mu$ has marginals $\mu_A, \mu_B$, we define $I_\infty(\mu)=\max_{(a, b)\in \cA\times \cB} \log_2 \frac{\mu(a, b)}{\mu_A(a)\mu_B(b)}$.

\begin{theorem}[\cite{JRS02}]\label{thm:substate}
Let $\eps>0$ be a positive parameter, and let $\cA, \cB$ be finite sets. For every distribution $\mu$ on $\cA\times \cB$, with mutual information bounded by $k$, there exists another distribution $\nu$, such that $\|\mu-\nu\|_{\operatorname{TV}}\leq \eps$ and $I_{\infty}(\nu)\leq 4(k+1)/\eps$.
\end{theorem}

Using this theorem, we give a proof of Theorem~\ref{thm:main_MI}.

\begin{proof}[Proof of Theorem~\ref{thm:main_MI}.]
Let us consider the measure $\nu$ such that $\|\mu-\nu\|_{\operatorname{TV}}\leq \eps/2$ and $I_\infty(\nu)\leq 8(k+1)/\eps$, which exists due to Theorem~\ref{thm:substate}. If the marginals of $\nu$ are $\nu_A$ and $\nu_B$, then we consider the product distribution $\nu_A\times \nu_B$ and we run the protocol from the proof of Theorem~\ref{thm:main_product} with the error parameter $\eps'= \eps 2^{-8(k+1)/\eps-1}$ and distribution $\nu_A\times \nu_B$.

The cost of this protocol is $O(\sqrt{n\log_2 1/\eps'})=O(\sqrt{n(k+1)/\eps})$. Moreover, if $S\subseteq 2^{[n]}\times 2^{[n]}$ is the set of pairs of inputs where it makes the mistake, we have that $S$ has measure at most $(\nu_A\times \nu_B)(S)\leq \eps'$. Since $I_\infty(\nu)\leq 8(k+1)/\eps$, this means that for any pair $(A, B)\in \cA\times \cB$ we have $\nu(A, B)\leq 2^{8(k+1)/\eps}\nu_A(A)\nu_B(B)$ and so $\nu(S)\leq 2^{8(k+1)/\eps}\eps'\leq \eps/2$.

Finally, since $\|\mu-\nu\|_{\operatorname{TV}}\leq \eps/2$, we have $|\mu(S)-\nu(S)|\leq \eps/2$, allowing us to conclude $\mu(S)\leq \eps$. Hence, we have a protocol of complexity at most $O(\sqrt{n(k+1)/\eps})$ making an error on at most an $\eps$-fraction of inputs, completing the proof.
\end{proof}

\medskip
\noindent
{\bf Acknowledgment.} The authors would like to thank Avi Wigderson for valuable insights and for bringing the paper of Babai, Frankl, and Simon to their attention.

\bibliographystyle{alpha}

\end{document}